\documentclass{tSYS2e}

\usepackage{color}
\theoremstyle{plain}
\newtheorem{theorem}{Theorem}

\theoremstyle{definition}
\newtheorem{remark}{Remark}

\begin{document}



\title{Adaptive Control of Uncertain Pure-feedback Nonlinear Systems}

\author{Mingzhe Hou$^{a}$$^{\ast}$\thanks{$^\ast$Corresponding author. Email: hithyt@hit.edu.cn
\vspace{6pt}} Zongquan Deng$^{b}$ and Guangren Duan$^{a}$ \\\vspace{6pt} $^{a}${\em{Center for Control Theory and Guidance Technology, Harbin
Institute of Technology, Harbin 150001, China}};
$^{b}${\em{School of Mechatronic Engineering, Harbin
Institute of Technology, Harbin 150001, China}}\\\received{Sep., 2016} }

\maketitle

\begin{abstract}
A novel adaptive control approach is proposed to solve the globally asymptotic state stabilization problem
for uncertain pure-feedback nonlinear systems which can be transformed into the pseudo-affine form.
The pseudo-affine pure-feedback nonlinear system under consideration is with non-linearly parameterised uncertainties
and possibly unknown control coefficients. Based on the parameter separation technique, a backstepping controller is designed by adopting the adaptive high gain idea. The rigorous stability analysis shows that the proposed controller could guarantee, for any initial system condition,
boundedness of the closed-loop signals and globally asymptotic stabilization of the state.
A numerical and a realistic examples are employed to demonstrate the effectiveness of the proposed control method.


\end{abstract}

\section{Introduction}
During the past several decades, control of nonlinear systems has received considerable attention
and a series of powerful control methods have been proposed, such as sliding mode control
\citep{Slotine1991}, feedback linearization \citep{Sastry1989},
backstepping \citep{Kanellakopoulos1991}, and so on. Among these methods, the backstepping technique,
which is a constructive, recursive, Lyapunov-based control design approach,
is one of the most popular control design tools to deal with a large
class of nonlinear systems with uncertainties, especially with unmatched
uncertainties. Nonlinear systems using the backstepping technique are of
lower triangular form which can be broadly classified into two kinds:
strict-feedback form and pure-feedback form. In the case of
strict-feedback systems, a great deal of progress has been achieved to
develop backstepping controllers. It is first proposed for
nonlinear systems with linearly parameterized uncertainties \citep
{Kanellakopoulos1991, Krstic1995}, and then extended to handle
non-linearly parameterized uncertainties \citep{Lin2002a, Lin2002b, Niu2005}. By
introducing Nussbaum functions, it is also applied to
nonlinear systems with unknown control directions \citep{Ye1999}.
The relevant engineering applications
have also been widely discussed, for example, from
wheeled mobile robots \citep{Mnif2005} to aircrafts \citep{Lungu 2013} to spacecrafts \citep{Ali2010}.

Compared with this progress, relatively fewer results are available for control
of pure-feedback systems. Pure-feedback systems, which have no
affine appearance of the state variables to be used as virtual controls
and/or the actual control, is more representative than
strict-feedback systems. Many practical systems are of pure-feedback form,
such as aircraft control systems \citep{Hunt1997}, biochemical processes \citep
{Krstic1995}, mechanical systems \citep{Ferrara2000}, and so on. Cascade and non-affine
properties of pure-feedback systems make it rather difficult to find
explicit virtual controls and the actual control. Therefore, control of
non-affine pure-feedback nonlinear systems is a meaningful and challenging
issue and has become a hot topic in the control field in recent
years (See \cite{Liu2016, Tong2016} and references therein). Most of these
results are approximation-based approaches which fully exploit the universal
approximation ability of neural networks or fuzzy logic systems.
By utilizing the Mean Value Theorem, the original non-affine
system is transformed to the quasi-affine form. Subsequently, the
backstepping technique is employed for the control design. Generally speaking,
it is difficult to obtain the explicit expressions of the ideal
virtual or actual controllers although their existence is guaranteed by the
Implicit Function Theorem. Hence, to approximate them,
estimators are constructed based on fuzzy logic systems \citep{Gao2013,
Li2015, Yu2013, Zhang2010} or  neural networks \citep{Shen2014, Sun2013,
Wang2006, Wang2011, Wang2013}. In \cite{Yoo2012} and \cite{Gao2012}, the singular perturbation
theory is also employed to estimate the ideal controllers.

Approximation-based approaches may suffer some problems. Take the approaches based on neural networks for example.
When the number of the neural network nodes increases to improve the approximation ability,
the number of adjustable parameters will become enormous. Accordingly, the online learning time will become very large \citep{Tong2016}.
In addition, the obtained results usually hold non-globally. To avoid suchlike problems, some approximation-free approaches are proposed
recently. In \cite{Liu2014}, a backstepping control algorithm is proposed
for a class of pure-feedback nonlinear systems by viewing $f_{i}(\bar{x}_{i+1})$ in stead of
$x_{i+1}$ as the virtual control variable and by adding anintegrator.
This result is further extended to the case where there exist linearly parameterised uncertainties \citep{Liu2016}.
In \cite{Tong2016}, by adopting the barrier Lyapunov function technique,
an adaptive control technique is developed for a class of pure-feedback
systems with linearly parameterised uncertainties and full state constraints.

In this paper, the globally asymptotic state stabilization (GASS) problem is
considered for a class of pure-feedback systems which can be written into the pseudo-affine form.
The  pseudo-affine pure-feedback system under consideration has non-linearly parameterised uncertainties
and possibly unknown control coefficients. For this kind of systems,
if the lower bound of every control coefficient is exactly known,
the GASR problem can be solved by the non-smooth or the smooth control schemes given in \cite{Lin2002a, Lin2002b}.
When their lower bounds are unknown, a possible way is to adopt and improve the adaptive control approach proposed in \cite{Sun2007}. If so, $2n$ adaptive laws are needed to estimate the unknown parameters.
In this paper, motivated by the work in \cite{Ye1999} and the high gain idea given in \cite{Lei2006},
a novel adaptive backstepping controller is proposed based on the parameter separation technique \citep{Lin2002a,Lin2002b}.
In the proposed method, only $n$ adaptive laws are needed. Since no estimators are needed to approximate the
ideal controllers, drawbacks of the approximation-based approaches can be avoided.
The rigorous stability analysis shows that the proposed controller could guarantee bounded closed-loop signals and globally
asymptotic state stabilization.

\section{Problem Formulation}
Consider pure-feedback nonlinear systems which can be written into the following pseudo-affine form
\begin{eqnarray}
\left\{
\begin{array}{l}
\dot{x}_{i} =\varphi _{i}(\bar{x}_{i},\theta )+g_{i}(\bar{x}_{i+1},\theta)x_{i+1}, (i=1,\cdots,n-1) \\
\dot{x}_{n} =\varphi _{n}(x,\theta )+g_n(x,\theta,u)u\\
y=x_1
\end{array}
\right.  \label{System0}
\end{eqnarray}
where $x=[
\begin{array}{ccc}
x_{1} & \cdots & x_{n}
\end{array}
]^{\mathrm{T}}$ is the state, $\bar{x}_{i}=[
\begin{array}{ccc}
x_{1} & \cdots & x_{i}
\end{array}
]^{\mathrm{T}}$, $u$ is the input, $y$ is the output, $\theta :\Re _{\geq 0}\rightarrow \Re ^{m}$ is a
bounded, uncertain time-varying piecewise continuous parameter or disturbance vector,
$\varphi_{i}:\Re^{i+m}\rightarrow \Re $ and $g_{i}:\Re^{i+1+m}\rightarrow \Re $ are
$\mathcal{C}^{1}$ functions, and $\varphi_{i}(0_{i\times 1},\theta)=0$.
For notational convenience, we introduce $x_{n+1}:=u$.

\begin{remark} { For general pure-feedback systems
\begin{equation}
\left\{
\begin{array}{l}
\dot{x}_{i}=f_{i}(\bar{x}_{i+1},\theta), (i=1,\cdots,n-1) \\
\dot{x}_{n}=f_{n}(x,\theta ,u)
\end{array}
\right.  \label{System1}
\end{equation}
where $f_{i}:\Re^{i+1+s}\rightarrow \Re $ are $\mathcal{C}^{2}$ functions satisfying $
f_{i}(0_{i\times 1},\theta,0)=0$, by the Mean Value Theorem, we
have
\begin{eqnarray}
f_{i}(\bar{x}_{i+1},\theta)=f_{i}(\bar{x}_{i},0,\theta)+\frac{
\partial f_{i}(\bar{x}_{i},\lambda _{i}x_{i+1},\theta)}{\partial x_{i+1}}
x_{i+1} \label{Dec1}
\end{eqnarray}
where $\lambda _{i}\in (0,1)$ and is a constant or a $\bar{x}_{i+1}$-dependent variable.
Define $\varphi _{i}(\bar{x}_{i},\theta)=f_{i}(\bar{x}_{i},0,\theta)$ and $
g_{i}(\bar{x}_{i+1},\theta)=\frac{\partial f_{i}(\bar{x}
_{i},\lambda _{i}x_{i+1},\theta)}{\partial x_{i+1}}$.
It is easy to know that $\varphi_{i}$ and $g_{i}$ are $\mathcal{C}^{1}$ functions, and $\varphi_{i}(0_{i\times 1},\theta)=0$.
Hence, system (\ref{System1}) can be transformed to the form of (\ref{System0}).}
\end{remark}

In this paper, we are interested in the case where $\varphi _{i}$ and $g_{i}$ satisfy the following two assumptions, respectively.

\textbf{Assumption 1:} There exist a set of unknown constants $
c_{i} $ and known $\mathcal{C}^{1}$ functions $\rho _{i}:\Re ^{i}\rightarrow \Re _{\geq
0}$ such that
\begin{eqnarray}
\left\vert\varphi _{i}(\bar{x}_{i},\theta )\right\vert\leq c_{i}\rho
_{i}(\bar{x}_{i})\sum_{j=1}^{i}\left\vert x_{j}\right\vert, (i=1,\cdots,n)  \label{Assumption1}
\end{eqnarray}
\textbf{Assumption 2:} The signs of $g_{i}$ are known.
Without loss of generality, we assume that they are all positive. In addition,
we assume that there exist a set of unknown constants $b_{i}>0$, $B_{i}>0$, and
known $\mathcal{C}^{1}$ functions $\phi _{i}:\Re ^{i+1}\rightarrow \Re _{> 0}$ such that
\begin{equation}
0<b_{i}\leq g_{i}(\bar{x}_{i+1},\theta)\leq B_{i}\phi _{i}(\bar{x}_{i+1}), (i=1,\cdots ,n-1)  \label{Assumption2}
\end{equation}
and $g_{n}(x,\theta,u)\geq b_{n}>0$.

The control problem to be solved is stated as follows.

\textbf{GASS Problem:} Consider the
uncertain pure-feedback system (\ref{System0}) under Assumptions 1 and 2.
Design a $\mathcal{C}^{1}$ state feedback controller $u$ such that all signals in
the resulting closed-loop system are bounded on $[0,\infty )$, furthermore,
globally asymptotic stabilization of the state $x$ is achieved, \emph{i.e.} $
\lim_{t\rightarrow \infty }x(t)=0$ for all $x(0)\in \Re ^{n}$.

\begin{remark}
The condition $g_{i}\geq b_{i}>0$ actually presents a sufficient global
controllability condition for system (\ref{System0}).
In addition, since $\varphi _{i}(\bar{x}_{i},\theta )$ is a $\mathcal{C}^{1}$ function,
according to \cite{Nijmeijer1990}, it can be rewritten as $\varphi _{i}(\bar{
x}_{i},\theta )=\sum\nolimits_{j=1}^{i} \xi _{ij}(\bar{x}_{i},\theta )x_{j}$ with $
\xi _{ij}$ being continuous functions. According to the parameter separation
technique introduced in \cite{Lin2002a} and \cite{Lin2002b}, it is reasonable to assume that,
$\forall j=1,\cdots,i$, $\left|\xi_{ij}(\bar{x}_{i},\theta )\right| \leq c_{i}\rho _{i}(\bar{x}_{i})$, which implies that
$\left\vert\varphi _{i}( \bar{x}_{i},\theta )\right\vert\leq c_{i}\rho _{i}(\bar{x}
_{i})\sum\nolimits_{j=1}^{i}\left\vert x_{j}\right\vert$, and
$g_{i}(\bar{x}_{i+1},\theta)\leq B_{i}\phi _{i}(\bar{x}_{i+1})$, where $c_{i}$ and $B_{i}$
are unknown parameters, and $\rho _{i}$ and $\phi _{i}$
are known smooth functions. In fact, assumptions like Assumptions 1 and 2 are frequently
used in the literature, such as \cite{Lin2002a}, \cite{Lin2002b}, \cite{Sun2007}, and so on.
It is worth pointing out that the control coefficients $g_{i}$ themselves could be unknown functions under Assumption 2.
\end{remark}

\section{Main results}

In this section, we shall first present an adaptive control scheme, and
subsequently we shall prove that it leads to the solution to the GASR
problem for system (\ref{System0}).

\subsection{Control Scheme}

The proposed adaptive control scheme is given in a step-by-step way as
follows.

\textbf{\emph{Step 0:}} Introduce the following coordinate transformation
\begin{eqnarray}
z_{1} &=&x_{1}  \label{Coordinates1} \\
z_{i} &=&x_{i}-\alpha _{i-1},i=2,\cdots ,n  \label{Coordinates2}
\end{eqnarray}
where $\alpha _{i},i=1,\cdots ,n-1$ are the virtual control laws to be determined.
For notational convenience, we introduce $\alpha_{n}:=u $.
Virtual control laws $\alpha _{i},i=1,\cdots ,n-1$ and the actual control law $\alpha_{n}$ ,{\it i.e.}, $u$,
are constructed as
\begin{eqnarray}
\alpha_{i}(\bar{x}_i, \bar{k}_i)&=&-\mu _{i}k_{i}\psi _{i}^{2}(\bar{z}_{i},\bar{k}
_{i-1})z_{i}, (i=1,\cdots ,n)  \label{Virtual}
\end{eqnarray}
where $k_{i},i=1,\cdots ,n$ are update laws given by
\begin{eqnarray}
\dot{k}_{i} &=&\gamma _{i}\psi _{i}^{2}(\bar{z}_{i},\bar{k}
_{i-1})z_{i}^{2}
\doteq\omega _{i}, k_{i}(0)=k_{i0}\in \Re _{>0}\label{Update}
\end{eqnarray}
where $\psi _{i}$ are $\mathcal{C}^{1}$ functions to be determined in the
following $i$th step, $\mu _{i}\in \Re _{>0}$ and $\gamma _{i}\in \Re _{>0}$ are design parameters.
Roughly speaking, by choosing bigger $\mu _{i}$ and $\gamma _{i}$, the convergence speed of the states can be improved.
However, this may lead to larger control magnitude. Therefore, a tradeoff is required when determining these design parameters in applications.

\textbf{\emph{Step 1:}} Start with
\begin{eqnarray}
\dot{z}_{1} &=&\dot{x}_{1}
=\varphi _{1}+g_{1}x_{2}
\label{z1dot}
\end{eqnarray}
Define $V_{1}=\frac{1}{2}z_{1}^{2}$. According to (\ref{Coordinates2})-(\ref{Update}) and Assumptions 1 and 2, the time derivative of $V_{1}$ is such that
\begin{eqnarray}
\dot{V}_{1} &=&z_{1}\dot{z}_{1}
\leq c_{1}\rho _{1}(z_{1})z_{1}^{2}+g_{1}z_{1}(z_{2}+\alpha _{1})  \nonumber \\
&\leq &-b_{1}\mu _{1}k_{1}\psi _{1}^{2}z_{1}^{2}+c_{1}\rho
_{1}(z_{1})z_{1}^{2}+z_{1}^{2}+B_{1}^{2}\phi _{1}^{2}(\bar{x}_{2})z_{2}^{2}\label{V1dot0}
\end{eqnarray}
Choose $\psi _{1}$ to be any $\mathcal{C}^{1}$ function satisfying
\begin{equation}
\rho _{1}(z_{1})+1\leq \psi _{1}(z_{1})  \label{Psi1}
\end{equation}
then
\begin{eqnarray}
\dot{V}_{1} &\leq &-b_{1}\mu _{1}k_{1}\psi _{1}^{2}z_{1}^{2}+\beta _{1}\psi
_{1}^{2}z_{1}^{2}+z_{1}^{2}+\sigma _{1}\phi _{1}^{2}(\bar{x}_{2})z_{2}^{2}
\nonumber \\
&=&-\frac{1}{\gamma _{1}}b_{1}\mu _{1}k_{1}\dot{k}_{1}+\frac{1}{\gamma _{1}}
\beta _{1}\dot{k}_{1}+z_{1}^{2}+\sigma _{1}\phi _{1}^{2}(\bar{z}
_{2},k_1)z_{2}^{2}  \label{V1dot}
\end{eqnarray}
where $\beta _{1} =c_{1}$ and $\sigma _{1} =B_{1}^{2}$
are unknown constants.

\textbf{\emph{Step }$\textbf{\emph{i}}$\textbf{\ }$(2\leq i\leq n)$}: The derivative of $
z_{i}$ is
\begin{eqnarray}
\dot{z}_{i} &=&\dot{x}_{i}-\dot{\alpha}_{i-1}
=\varphi _{i}+g_{i}x_{i+1}
-\sum_{j=1}^{i-1}\left[ \frac{\partial
\alpha _{i-1}}{\partial x_{j}}\left( \varphi _{j}+g_{j}x_{j+1}\right) +\frac{
\partial \alpha _{i-1}}{\partial k_{j}}\omega _{j}\right]  \label{zidot}
\end{eqnarray}
Define $V_{i}=\frac{1}{2}z_{i}^{2}$. Bearing (\ref{Coordinates2})-(\ref{Update}) and Assumptions 1 and 2 in mind, we can obtain
\begin{eqnarray}
\dot{V}_{i} &=&z_{i}\dot{z}_{i}
=z_{i}\varphi _{i}+g_{i}z_{i}x_{i+1}
-z_{i}\sum_{j=1}^{i-1}\left[ \frac{\partial \alpha _{i-1}}{
\partial x_{j}}\left( \varphi _{j}+g_{j}x_{j+1}\right) +\frac{\partial
\alpha _{i-1}}{\partial k_{j}}\omega _{j}\right]   \nonumber \\
&\leq &g_{i}z_{i}x_{i+1}+c_{i}\rho _{i}(\bar{x}_{i})\left\vert
z_{i}\right\vert \sum_{j=1}^{i}\left\vert x_{j}\right\vert\nonumber\\
&&+\left\vert z_{i}\right\vert
\sum_{j=1}^{i-1}\left[ \left\vert
\frac{\partial \alpha _{i-1}}{\partial x_{j}}\right\vert \left( c_{j}\rho
_{j}(\bar{x}_{j})\sum_{l=1}^{j}\left\vert x_{l}\right\vert
+g_{j}\left\vert x_{j+1}\right\vert \right)+\gamma _{j}\left\vert \frac{\partial \alpha _{i-1}}{\partial k_{j}}
\right\vert \psi _{j}^{2}(\bar{z}_{j},\bar{k}_{j-1})z_{j}^{2}\right]   \nonumber\\
&\leq &g_{i}z_{i}x_{i+1}+\vartheta _{i}\left\vert z_{i}\right\vert \rho _{i}(
\bar{x}_{i})\sum_{j=1}^{i}\left\vert x_{j}\right\vert \nonumber \\
&&+\vartheta _{i}\left\vert z_{i}\right\vert
\sum_{j=1}^{i-1}\left[
\left\vert \frac{\partial \alpha _{i-1}}{\partial x_{j}}\right\vert \left(
\rho _{j}(\bar{x}_{j})\sum_{l=1}^{j}\left\vert
x_{l}\right\vert +\phi _{j}(\bar{x}_{j+1})\left\vert x_{j+1}\right\vert
\right)+\gamma _{j}\left\vert \frac{\partial \alpha _{i-1}}{\partial k_{j}}
\right\vert \psi _{j}^{2}(\bar{z}_{j},\bar{k}_{j-1})z_{j}^{2}\right]   \label{Vidot0}
\end{eqnarray}%
where $\vartheta _{i}=\max (1,c_{1},c_{2},\cdots ,c_{i},B_{1},\cdots ,B_{i})$.
According to (\ref{Coordinates2}) and (\ref{Virtual}), we have that, for $2\leq q\leq i$,
\begin{eqnarray}
\left\vert x_{q}\right\vert \leq \left\vert z_{q}\right\vert +\mu _{q-1}k_{q-1}\psi _{q-1}^{2}(\bar{z}
_{q},\bar{k}_{q-2})\left\vert z_{q-1}\right\vert  \label{xq}
\end{eqnarray}
Hence,
\begin{equation}
\dot{V}_{i}\leq g_{i}z_{i}x_{i+1}+\vartheta _{i}\left\vert z_{i}\right\vert
\eta _{i}(\bar{z}_{i},\bar{k}_{i-1})\sum_{j=1}^{i}\left\vert
z_{j}\right\vert  \label{Vidot2}
\end{equation}
where $\eta _{i}(\bar{z}_{i},\bar{k}_{i-1})$ is a known nonnegative-valued
function. Now we choose $\psi _{i}$ to be any $\mathcal{C}^{1}$ function satisfying
\begin{equation}
\psi _{i}(\bar{z}_{i},\bar{k}_{i-1})\geq \eta _{i}(\bar{z}_{i},\bar{k}
_{i-1})+\phi _{i-1}(\bar{z}_{i},\bar{k}_{i-1})+1  \label{psii}
\end{equation}
Then, we can obtain
\begin{eqnarray}
\dot{V}_{i} &\leq &g_{i}z_{i}(z_{i+1}+\alpha _{i})+\vartheta _{i}\left\vert
z_{i}\right\vert \eta _{i}(\bar{z}_{i},\bar{k}_{i-1})\sum_{j=1}^{i}
\left\vert z_{j}\right\vert   \nonumber \\
&\leq &-b_{i}\mu _{i}k_{i}\psi _{i}^{2}(\bar{z}_{i},\bar{k}
_{i-1})z_{i}^{2}+z_{i}^{2}+\sigma _{i}\phi _{i}^{2}(\bar{x}
_{i+1})z_{i+1}^{2}+i\vartheta _{i}^{2}\psi _{i}^{2}(\bar{x}_{i},\bar{k}
_{i-1})z_{i}^{2}+\sum_{j=1}^{i}z_{j}^{2}  \nonumber \\
&\leq &-b_{i}\mu _{i}k_{i}\psi _{i}^{2}(\bar{z}_{i},\bar{k}
_{i-1})z_{i}^{2}+\beta _{i}\psi _{i}^{2}(\bar{z}_{i},\bar{k}
_{i-1})z_{i}^{2}+\sigma _{i}\phi _{i}^{2}(\bar{z}_{i+1},\bar{k}
_{i})z_{i+1}^{2}+\sum_{j=1}^{i}z_{j}^{2}  \nonumber \\
&=&-\frac{1}{\gamma _{i}}b_{i}\mu _{i}k_{i}\dot{k}_{i}+\frac{1}{\gamma _{i}}
\beta _{i}\dot{k}_{i}+\sigma _{i}\phi _{i}^{2}(\bar{z}_{i+1},\bar{k}
_{i})z_{i+1}^{2}+\sum_{j=1}^{i}z_{j}^{2}  \label{Vidot}
\end{eqnarray}
where $z_{n+1}=x_{n+1}-\alpha _{n}=0$, $\beta _{i} =i\vartheta _{i}^{2}+1$ and $\sigma _{i} =B_{i}^{2}$
are unknown constants.

\subsection{Stability Analysis}

\begin{theorem}
Suppose that Assumptions 1 and 2 are satisfied and that the above-proposed
design procedure is applied to system (\ref{System0}), then for all $x(0)\in
\Re ^{n}$ and fixed $k_{i}(0)=k_{i0}>0$, $i=1,\cdots ,n$, all signals in
the resulting closed-loop system are bounded on $[0,\infty )$, furthermore, $
\lim_{t\rightarrow \infty }x(t)=0$ and $\lim_{t\rightarrow \infty }k_{i}(t)=
k_{i\infty}\in \Re _{>0}$.
\end{theorem}

\begin{proof}
Due to the smoothness of the proposed robust control, the solution of the
closed-loop system has a maximum interval of existence $[0,t_{f})$ where $
t_{f}\in \Re _{>0}$. We will first prove the boundedness of all state
variables on the interval $[0,\infty )$, and then prove the convergence of $
x(t)$ and $k_{i}(t)$. To this end, we define
\begin{equation}
V=\sum_{i=1}^{n}V_{i}=\frac{1}{2}\sum_{i=1}^{n}z_{i}^{2}\geq 0
\label{V}
\end{equation}
It follows from (\ref{Psi1}), (\ref{V1dot}), (\ref{psii}), (\ref{Vidot}) and
the fact $z_{n+1}=0$ that
\begin{eqnarray}
\dot{V} &=&\sum_{i=1}^{n}\dot{V}_{i}
\leq-\sum_{i=1}^{n}\frac{1}{\gamma _{i}}b_{i}\mu _{i}k_{i}\dot{k}
_{i}+\sum_{i=1}^{n}\frac{1}{\gamma _{i}}\beta _{i}\dot{k}_{i}
+\sum_{i=1}^{n-1}\sigma _{i}\phi _{i}^{2}(\bar{z}_{i+1},\bar{k}
_{i})z_{i+1}^{2}+\sum_{i=1}^{n}\sum_{j=1}^{i}z_{j}^{2}\nonumber \\
&\leq &-\sum_{i=1}^{n}\frac{1}{\gamma _{i}}b_{i}\mu _{i}k_{i}\dot{k}
_{i}+\sum_{i=1}^{n}\frac{1}{\gamma _{i}}\beta _{i}\dot{k}_{i}
+\sum_{i=1}^{n-1}\sigma _{i}\psi _{i+1}^{2}(\bar{z}_{i+1},\bar{k
}_{i})z_{i+1}^{2}
+\sum_{i=1}^{n}\sum_{j=1}^{i}\psi
_{j}^{2}(\bar{z}_{j},\bar{k}_{j-1})z_{j}^{2}  \nonumber \\
&\leq &-\sum_{i=1}^{n}\frac{1}{\gamma _{i}}b_{i}\mu _{i}k_{i}\dot{k}
_{i}+\sum_{i=1}^{n}\frac{1}{\gamma _{i}}\beta _{i}\dot{k}_{i}
+\sum_{i=1}^{n-1}\sigma _{i}\frac{1}{\gamma _{i+1}}\dot{k}
_{i+1}+\sum_{i=1}^{n}\sum_{j=1}^{i}\frac{1}{\gamma _{j}}
\dot{k}_{j}  \nonumber \\
&\leq &-\sum_{i=1}^{n}\frac{1}{\gamma _{i}}b_{i}\mu _{i}k_{i}\dot{k}
_{i}+\varepsilon \sum_{i=1}^{n}\frac{1}{\gamma _{i}}\dot{k}
_{i}+\varepsilon \sum_{i=1}^{n}\frac{1}{\gamma _{i}}\dot{k}_{i}
\nonumber \\
&=&-\sum_{i=1}^{n}\frac{1}{\gamma _{i}}b_{i}\mu _{i}k_{i}\dot{k}
_{i}+2\varepsilon \sum_{i=1}^{n}\frac{1}{\gamma _{i}}\dot{k}_{i}
\label{Vdot}
\end{eqnarray}
where $
\varepsilon =\max \left( \beta _{1},\cdots ,\beta _{n},n,n-1+\sigma
_{1},\cdots ,1+\sigma _{n-1}\right)$
is an unknown constant. Integrating inequality (\ref{Vdot}) gives, $\forall t\in \lbrack 0,t_{f})$,
\begin{eqnarray}
V(t)&\leq& -\frac{1}{2}\sum_{i=1}^{n}\frac{1}{\gamma _{i}}
b_{i}\mu _{i}k_{i}^{2}(t)
+2\varepsilon \sum_{i=1}^{n}\frac{1}{
\gamma _{i}}k_{i}(t)+C_{0}  \label{IntV}
\end{eqnarray}
where $C_{0}$ is a constant depending on initial data. This implies that $
k_{i}$ are bounded on $[0,t_{f})$. Otherwise, on taking limit as $
t\rightarrow t_{f}$, the right side of the above inequality would diverge to $
-\infty $, which would yield a contradiction with $V(t)\geq 0$. It immediately follows from (\ref{IntV}) that $V(t)$, and in
turn, $z_{i}(t)$ are bounded. According to (\ref{Coordinates1})-(\ref
{Virtual}), $x_{i}(t)$ are bounded. Therefore, all the state variables
of the closed-loop system are bounded on the interval $[0,t_{f})$, hence, $
t_{f}=\infty $. As a result, the control $u$ which depends on $
x_{i}$ and $k_{i}$ is bounded, which means that $\dot{x}_{i}(t)$ is
bounded. Furthermore, $\dot{k}_{i}(t)$ is bounded according to (\ref
{Update}). Hence, by (\ref{z1dot}) and (\ref{zidot}), $\dot{
z}_{i}(t)$ is bounded.
{ Since $k_{i}(t)$ is bounded on $[0,+\infty)$ and $\dot{k}_{i}\geq 0$,
there exist constants $k_{i\infty}\in \Re _{>0}$ such that $\lim_{t\rightarrow \infty }k_{i}(t)=k_{i\infty}$.
From (\ref{Psi1}) and (\ref{psii}), we have $\psi_i \geq 1$. This, together with (\ref{Update}), implies that $z_{i}^{2}(t) \leq \frac{k_i(t)}{\gamma_i}$. Thus,
\begin{eqnarray}
&&\lim_{t\rightarrow \infty}\int\nolimits_{0}^{t}z_{i}^{2}(\tau )\mathrm{d}\tau
\leq \lim_{t\rightarrow \infty}\int\nolimits_{0}^{t}\frac{\dot{k}_{i}(\tau )}{\gamma_i}\mathrm{d}\tau
=\frac{k_{i\infty}-k_{i0}}{\gamma_i}
\label{Intofzs}
\end{eqnarray}
Therefore, by using Barbalat's Lemma \citep{Tao1997,Hou2010}, we have $\lim_{t\rightarrow \infty }z_i(t)=0$, that is, $\lim_{t\rightarrow \infty }z(t)=0$.
Consequently, from (\ref{Coordinates1})-(\ref{Virtual}) and the boundedness of $k_{i}(t)$, we can get that
$\lim_{t\rightarrow \infty }x(t)=0$.
This completes the proof.
}
\end{proof}

\begin{remark}
The proposed control algorithm is also available if $\phi_i$ and $g_i$ are functions of $x$. Of course,
$\rho _{i}$ and $\phi _{i}$ are still functions of $\bar{x}_{i}$ and $\bar{x}_{i+1}$, respectively.
In this case, system (\ref{System0}) is beyond the lower triangular form.
\end{remark}

\begin{remark}
The intuitive explanation of the proposed method is given as follows. As long as $z_i\neq 0$, $k_i$ will keep on growing.
When $k_i$ become large enough, the uncertainties will be dominated completely and $z_i$ will converge to zero eventually.
In practical applications, if necessary, we could add the following modification to the update law (\ref{Update}):
$\dot{k}_i=0$ if $|z_i|<\delta$, where $\delta$ characterizes the tolerable error range.
In this case, if $|z_i|\geq \delta$, $k_i$ will grow.
When $k_i$ become large enough, the uncertainties will be dominated fully
and $|z_i|$ will decrease ultimately and be kept within the tolerable range. This way could prevent
$k_i$ from increasing unboundedly and reinforce the closed-loop robustness.
\end{remark}

\section{Examples}
In this section, we will give two examples to show the effectiveness
of the obtained results.

\subsection{Numerical Example}
Let us consider the globally asymptotic state regulation of the following system
\begin{equation}
\left\{
\begin{array}{l}
\dot{x}_{1}=\theta _{1}\left(x_{1}+x_{2}+\frac{x_{2}^{3}}{5}\right) \\
\dot{x}_{2}=\theta _{2}\left(x_{1}x_{2}+u+\frac{u^{3}}{7}\right)
\end{array}
\right.   \label{ExaSys}
\end{equation}
where $\theta _{i}$, $i=1,2$ are uncertain time-varying piecewise continuous parameters belonging to the interval $[\underline{\theta}_i, \overline{\theta}_i]$
with $\underline{\theta}_i$ and $\overline{\theta}_i$ being unknown positive constants. System (\ref{ExaSys}) with $\theta _{i}$=1 is
frequently discussed in the literature on control of pure-feedback systems.
To written system (\ref{ExaSys}) into the form of (\ref{System0}), we can set $\varphi
_{1}=\theta _{1}x_{1}$, $g_{1}=\theta _{1}\left( 1+\frac{x_{2}^{2}}{5}
\right) $, $\varphi _{2}=\theta _{2}x_{1}x_{2}$ and $g_{2}=\theta _{2}\left(
1+\frac{u^{2}}{7}\right) $. Let $b_{i}=\underline{\theta}_i$, $c_{i}=B_{i}=\overline{\theta}_i$, where $i=1,2$, $\rho _{1}=1
$, $\phi _{1}=1+\frac{x_{2}^{2}}{5}$, $\rho
_{2}=\frac{1+x_{1}^{2}+x_{2}^{2}}{4}$ and $\phi _{2}=1+\frac{u^{2}}{7}$,
then it is easy to check that Assumptions 1 and 2 are satisfied.

By the design procedure given in Section 3, define $z_{1}=x_{1}$ and $
z_{2}=x_{2}-\alpha _{1}$, the virtual controller and the actual controller can
be respectively constructed as
\begin{eqnarray*}
\alpha _{1} =-\mu _{1}k_{1}\psi _{1}z_{1},\ \
u =-\mu _{2}k_{2}\psi _{2}z_{2}
\end{eqnarray*}
with update laws
\begin{eqnarray*}
\dot{k}_{1} =\gamma _{1}\psi _{1}^{2}z_{1}^{2},\ \
\dot{k}_{2} =\gamma _{2}\psi _{2}^{2}z_{2}^{2}
\end{eqnarray*}
where $\psi _{1}=2$ and $\psi _{2}=(1+2\mu _{1}k_{1})\rho _{2}+8\gamma _{1}\mu_{1}z_{1}^{2}+\phi _{1}+1
+2\left(1+\phi _{1}+2\mu _{1}k_{1}\phi _{1}\right) \mu _{1}k_{1}$. For simulation, we set $\theta _{1}=1$
, $\theta _{2}=2$, $x_{10}=-2$, $x_{20}=3$, and for $i=1,2$, $k_{i0}=0.01$, $\mu _{i}=\gamma _{i}=0.2$.
The simulation results are given in Fig.\ref{fig1}.

\begin{figure}[!htb]
  \centering
  \includegraphics[width=0.9\hsize]{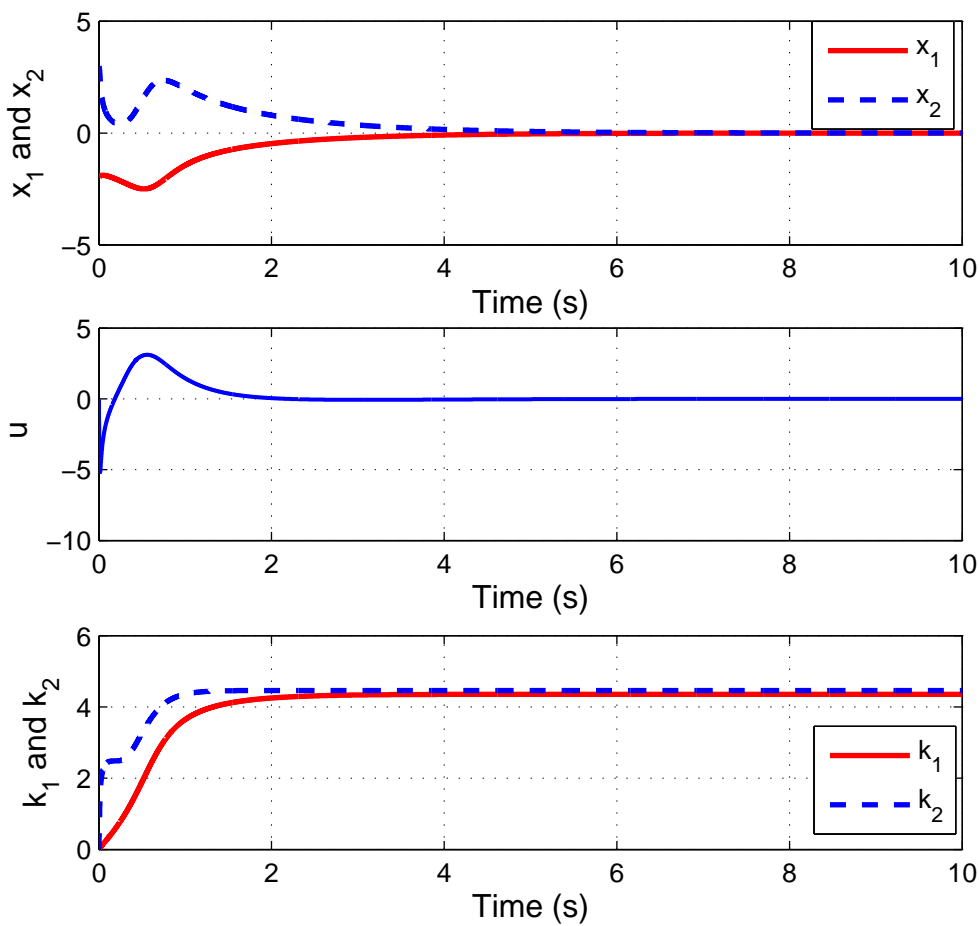}
  \caption{Simulation results for the numerical example}
  \label{fig1}
\end{figure}

\subsection{Realistic Example}

Consider the roll control of axially symmetric Skid-To-Turn (STT)
missiles, for which the simplified mathematical model with actuator dynamics is given by
\begin{eqnarray}
\left\{
\begin{array}{l}
\dot{\gamma}=\omega _{x} \\
\dot{\omega}_{x}=\frac{M_{x}}{J_{x}} \\
\dot{\delta}_{x}=\frac{\delta _{xc}-\delta _{x}}{\tau_a}
\end{array}
\right.  \label{STT}
\end{eqnarray}
where $\gamma$, $\omega _{x}$, $\delta _{x}$ and $\delta _{xc}$ are the roll angle, the rotating rate along the
roll axis, the aileron deflection angle and the aileron deflection angle
command to be determined, respectively. $J_{x}$, $\tau_a$ and $M_{x}$ are the moment of inertia about the roll axis, the time constant of the actuator and the roll moment, respectively. The mathematical expression of $M_{x}$ is given
by $M_{x}=\frac{1}{2}\rho V^{2}slm_{x}$, where $\rho$, $s$, $l$ and $m_{x}$ are the air density, the reference area, the reference length, and  the roll
moment coefficient. Roughly speaking, $m_{x}$ can be viewed as a smooth
nonlinear function of $\alpha$ (angle of attack), $\beta $ (angle of sideslip) and $\delta_{x}$ at some operating point
(Please see \cite{Siouris2004} and \cite{Hou2013}). When designing the roll controller for STT missiles, a
frequently used way is ignoring the impacts from the pitch and the yaw
channels. As a result, $m_{x}$ can be viewed as a smooth nonlinear function of $
\delta_{x}$, denoted by $m_{x}(\delta _{x})$. The nominal value of $
m_{x}(\delta _{x})$ can be determined by experiments and theoretical
calculation. For axially symmetric missiles, $m_{x}(0)=0$. Hence, by the
Mean Value Theorem, we have
\begin{eqnarray}
m_{x}(\delta _{x})&=&m_{x}(\delta _{x})-m_{x}(0)\nonumber \\
&=&\frac{\partial m_{x}(\lambda
\delta _{x})}{\partial \delta _{x}}\delta _{x}\doteq m_{x}^{\delta _{x}}(\lambda
\delta _{x})\delta _{x}  \label{Coefficient}
\end{eqnarray}
where $\lambda \in (0,1)$ and is $\delta_x$-dependent.
From (\ref{STT}) and (\ref{Coefficient}), we have
\begin{eqnarray}
\left\{
\begin{array}{l}
\dot{\gamma}=\omega _{x} \\
\dot{\omega}_{x}=\frac{\rho V^{2}slm_{x}^{\delta _{x}}(\lambda \delta _{x})}{
2J_{x}}\delta _{x} \\
\dot{\delta}_{x}=-\frac{1}{\tau_a}\delta _{x}+\frac{1}{\tau_a}\delta _{xc}
\end{array}
\right.  \label{STT1}
\end{eqnarray}

The objective is to design the aileron deflection angle command $\delta
_{xc} $ near some operating point such that all the states of the closed-loop system
converge to zero asymptotically.

Define
\begin{eqnarray*}
&&x_{1}=\gamma,~x_{2}=\omega _{x},~x_{3}=\delta _{x},~u=\delta _{xc},\nonumber\\
&&\theta_1=\frac{\rho V^{2}sl}{2J_{x}},~\theta_2=\frac{1}{\tau_a},\nonumber\\
&&g_{2}(\delta_{x},\theta_1)=\theta_1m_{x}^{\delta _{x}}(\lambda\delta _{x}),\nonumber\\
&&\varphi _{3}(\delta _{x},\theta_2)=-\theta_2\delta _{x},~g_{3}(\theta_2)=\theta_2,
\end{eqnarray*}
then we have
\begin{eqnarray}
\left\{
\begin{array}{l}
\dot{x}_{1}=x_{2} \\
\dot{x}_{2}=g_{2}(x_{3},\theta_1)x_{3} \\
\dot{x}_{3}=\varphi _{3}(x_{3},\theta_2)+g_{3}(\theta_2)u
\end{array}
\right.   \label{STT2}
\end{eqnarray}

Generally speaking, the function $m_{x}^{\delta _{x}}(\lambda \delta
_{x})=\frac{\partial m_{x}(\lambda \delta _{x})}{\partial \delta _{x}}$ are
unknown, but by experience, we can assume that
\begin{eqnarray*}
0<m\leq m_{x}^{\delta _{x}}(\lambda \delta _{x})\leq M \xi(\delta _{x})
\end{eqnarray*}
where $m$ and $M$ are two unknown constants,
and $\xi(\delta _{x})$ is a positive smooth function determined by the data obtained from experiments and theoretical
calculation. Near the operation point, $\theta_1$ continuously varies within some interval $[\theta_{1m}, \theta_{1M}]$,
where $\theta_{1m}$ and $\theta_{1M}$ are two positive constants, of which the values are not necessary to be known.
In practice, it is not easy to determine the exact value of the time constant of the actuator $\tau_a$.
But we know that it satisfies $\tau_a<0$, that is, $\theta_2<0$.

According to the above analysis, it is not easy to check that Assumptions 1 and 2 are satisfied
if we set $\rho _{1}=0$, $\phi_{1}=1$, $\rho _{2}=0$, $\phi_{2}=\xi$, $\rho_{3}=1$, and $\phi_{3}=1$.
Therefore, we could certainly design the aileron deflection angle command by strictly complying with the procedure given in Section 3.
However, we could simplify the the design procedure due to the special structure of (\ref{STT2}).

\textbf{\emph{Step 0:}} Define
\begin{eqnarray}
z_{1}=x_{1},~z_{2}=x_{2}-\alpha _{1},~z_{3}=x_{3}-\alpha _{2}
\end{eqnarray}
where $\alpha _{i}$ are virtual control laws to be determined.

\textbf{\emph{Step 1:}} Choose
\begin{eqnarray}
\alpha _{1}=-k_1z_1
\end{eqnarray}
where $k_1>0$ is design parameter, then one can obtain
\begin{eqnarray*}
\dot{z}_{1}=-k_1z_1+z_{2}
\end{eqnarray*}
Define $V_{1}=\frac{1}{2}z_1^2$, then one has
\begin{eqnarray*}
\dot{V}_{1}=-k_1z_1^2+z_1z_{2}\leq-\left(k_1-\frac{\epsilon}{4}\right)z_1^2+\frac{1}{\epsilon}z_{2}^2
\end{eqnarray*}
where $\epsilon>0$ is a constant to be determined.

\textbf{\emph{Step 2:}} Choose
\begin{eqnarray}
\alpha_{2}=-\mu _{2}k_{2}z_{2}
\end{eqnarray}
where $k_{2}$ is the update law given by
\begin{eqnarray}
\dot{k}_{2} =\gamma _{2}z_{2}^{2}, k_{2}(0)=k_{20}\in \Re _{>0}
\end{eqnarray}
where $\mu _{2}\in \Re _{>0}$ and $\gamma _{2}\in \Re _{>0}$ are design parameters. Then one has
\begin{eqnarray*}
\dot{z}_{2}=-g_{2}\mu _{2}k_{2}z_{2}+g_{2}z_3-k_1^2z_1+k_1z_{2}
\end{eqnarray*}
Define $V_{2}=\frac{1}{2}z_2^2$, then one has
\begin{eqnarray*}
&&\dot{V}_{2}=-g_{2}\mu _{2}k_{2}z_{2}^2+g_{2}z_2z_3-k_1^2z_1z_2+k_1z_{2}^2\nonumber\\
&&\leq-\frac{b_{2}\mu _{2}}{\gamma _{2}}k_{2}\dot{k}_{2}+B_{2}^2\phi _{2}^2z_3^2+\frac{\epsilon}{4}z_1^2+\left(\frac{k_1^4}{\epsilon}+k_1+\frac{1}{4}\right)z_{2}^2
\end{eqnarray*}

\textbf{\emph{Step 3:}} Choose
\begin{eqnarray}
u=-\mu _{3}k_{3}\psi_3^2z_{3}
\end{eqnarray}
where $k_{3}$ is the update law given by
\begin{eqnarray}
\dot{k}_{3} =\gamma _{3}\psi_3^2z_{3}^{2}, k_{3}(0)=k_{30}\in \Re _{>0}
\end{eqnarray}
where $\psi_3$ is a smooth positive function to be determined, $\mu _{3}\in \Re _{>0}$ and $\gamma _{3}\in \Re _{>0}$ are design parameters. Then one has
\begin{eqnarray*}
\dot{z}_{3}&=&-\theta_2(z_3+\alpha_2)-\theta_2\mu _{3}k_{3}\psi_3^2z_{3}+\mu _{2}\dot{k}_{2}z_{2}\nonumber\\
&&-g_{2}\mu _{2}^2k_{2}^2z_{2}+g_{2}\mu _{2}k_{2}z_3-\mu _{2}k_1^2 k_{2}z_1+\mu _{2}k_1k_{2}z_{2}
\end{eqnarray*}
Define $V_{3}=\frac{1}{2}z_3^2$, then one has
\begin{eqnarray*}
\dot{V}_{3} &=&\theta _{2}\mu _{2}k_{2}z_{2}z_{3}-\theta
_{2}z_{3}^{2}-\theta _{2}\mu _{3}k_{3}\psi _{3}^{2}z_{3}^{2}+\mu _{2}\dot{k}
_{2}z_{2}z_{3}  \nonumber \\
&-&g_{2}\mu _{2}^{2}k_{2}^{2}z_{2}z_{3}+g_{2}\mu _{2}k_{2}z_{3}^{2}-\mu
_{2}k_{1}^{2}k_{2}z_{1}z_{3}+\mu _{2}k_{1}k_{2}z_{2}z_{3}  \nonumber \\
&\leq &z_{2}^{2}+\frac{\epsilon }{4}z_{1}^{2}-\theta _{2}\mu _{3}k_{3}\psi
_{3}^{2}z_{3}^{2}+\left( \theta _{2}^{2}\mu _{2}^{2}k_{2}^{2}+\mu _{2}^{2}
\dot{k}_{2}^{2}\right.   \nonumber \\
&+&\left.B_{2}^{2}\phi _{2}^{2}\mu _{2}^{4}k_{2}^{4}+\frac{1}{4}
B_{2}^{2}\phi _{2}^{2}\mu _{2}^{2}k_{2}^{2}+\frac{\mu
_{2}^{2}k_{1}^{4}k_{2}^{2}}{\epsilon }+\mu_{2}^{2}k_{1}^{2}k_{2}^{2}+1\right) z_{3}^{2}  \nonumber
\end{eqnarray*}
where $\varpi =\max \{B_{2}^{2},B_{2},1\}$.

Define $V=V_1+V_2+V_3$, then one has
\begin{eqnarray*}
\dot{V} &\leq &-\left( k_{1}-\frac{3\epsilon }{4}\right) z_{1}^{2}-\frac{
b_{2}\mu _{2}}{\gamma _{2}}k_{2}\dot{k}_{2}+\varepsilon _{1}\dot{k}_{2}-
\frac{\theta _{2}\mu _{3}}{\gamma _{3}}k_{3}\dot{k}_{3} \\
&&+\varepsilon _{2}\gamma_3\left( \mu _{2}^{2}k_{2}^{2}+\mu _{2}^{2}\dot{k}
_{2}^{2}+\phi _{2}^{2}\mu _{2}^{4}k_{2}^{4}+\frac{1}{4}\phi _{2}^{2}\mu
_{2}^{2}k_{2}^{2}\right.  \\
&&+\left. \frac{\mu _{2}^{2}k_{1}^{4}k_{2}^{2}}{\epsilon }+\mu
_{2}^{2}k_{1}^{2}k_{2}^{2}+\phi _{2}^{2}+1\right) z_{3}^{2}
\end{eqnarray*}
where $\varepsilon_{2}=\frac{1}{\gamma _{2}}\left( \frac{k_{1}^{4}+1}{\epsilon }+k_{1}+\frac{5}{4}\right)$,  $\varepsilon _{3}=\frac{1}{\gamma_3}\max \{B_{2}^{2},\theta _{2}^{2},1\}$. Choose $\psi _{3}=\mu _{2}k_{2}+\mu _{2}\dot{k}_{2}+\phi _{2}\mu
_{2}^{2}k_{2}^{2}+\frac{1}{2}\phi _{2}\mu _{2}k_{2}+\frac{\mu
_{2}k_{1}^{2}k_{2}}{\sqrt{\epsilon }}+\mu _{2}k_{1}k_{2}+\phi _{2}+1$ and $\varepsilon_1=\left(k_{1}-\frac{3\epsilon }{4}\right)>0$, then one has
\begin{eqnarray*}
\dot{V}\leq -\varepsilon_1z_{1}^{2}-\frac{
b_{2}\mu _{2}k_{2}\dot{k}_{2}}{\gamma _{2}}+\varepsilon _{2}\dot{k}_{2}-
\frac{\theta _{2}\mu _{3}k_{3}\dot{k}_{3}}{\gamma _{3}}+\varepsilon _{3}\dot{k}_{3}
\end{eqnarray*}
From this, we can conclude that
\begin{eqnarray*}
V(t)+\varepsilon _{1}\int_{0}\nolimits^{t}z_{1}^{2}(\tau )\mathrm{d}\tau &\leq& -\frac{%
b_{2}\mu _{2}k_{2}^{2}(t)}{2\gamma _{2}}-\frac{\theta _{2}\mu
_{3}k_{3}^{2}(t)}{2\gamma _{3}}\nonumber\\
&&+\varepsilon _{2}k_{2}(t)+\varepsilon
_{3}k_{3}(t)+C_{0}
\end{eqnarray*}
This means that
\begin{eqnarray*}
V(t)\leq -\frac{
b_{2}\mu _{2}k_{2}^{2}(t)}{2\gamma _{2}}-\frac{\theta _{2}\mu
_{3}k_{3}^{2}(t)}{2\gamma _{3}}+\varepsilon _{2}k_{2}(t)+\varepsilon
_{3}k_{3}(t)+C_{0}
\end{eqnarray*}
and
\begin{eqnarray*}
\varepsilon _{1}\int_{0}\nolimits^{t}z_{1}^{2}(\tau )\mathrm{d}\tau &\leq& -\frac{
b_{2}\mu _{2}k_{2}^{2}(t)}{2\gamma _{2}}-\frac{\theta _{2}\mu
_{3}k_{3}^{2}(t)}{2\gamma _{3}}+\varepsilon _{2}k_{2}(t)\nonumber\\
&&+\varepsilon_{3}k_{3}(t)+C_{0}
\end{eqnarray*}
Therefore, similar to the the analysis given in Section 3.2, we can firstly use the former inequality to prove that
all signals in the resulting closed-loop system are bounded on $[0,\infty )$, and
$\lim_{t\rightarrow \infty }x_i(t)=0$, $\lim_{t\rightarrow \infty }k_{i}(t)=
k_{i\infty}\in \Re _{>0}$, $i=2,3$. Then, from the latter inequality, we have that
$\int_{0}\nolimits^{t}z_{1}^{2}(\tau )\mathrm{d}\tau$ is bounded on $[0,\infty )$. Since
$\int_{0}\nolimits^{t}z_{1}^{2}(\tau )\mathrm{d}\tau$ is a monotonic increasing function,
we have that $\lim_{t\rightarrow\infty}\int_{0}\nolimits^{t}z_{1}^{2}(\tau )\mathrm{d}\tau$.
So we can further use Barbalat's lemma to prove that $\lim_{t\rightarrow \infty }x_1(t)=0$.

The structure parameters of some missile (Please see \cite{Hou2013}) are $s=0.42$ $\mathrm{m}^{2}$, $L=0.68$
$\mathrm{m}$, $J_{x}=100$ $\mathrm{kg}\cdot \mathrm{m}^{2}$, $\tau_a=0.01$, and its aerodynamic
parameter $m_{x}^{\delta _{x}}(\lambda \delta _{x})$ at the operating point
with speed $200$ $\mathrm{m}/\mathrm{s}$ and height $5$ $\mathrm{km}$ ($\rho =0.7361$
$\mathrm{kg}/\mathrm{m}^{3}$) satisfies $m\leq m_{x}^{\delta _{x}}(\lambda \delta
_{x})\leq M$, where $m$ and $M$ are two unknown postive constants. Hence, we could choose $\xi=1$. For
numerical simulation, we set the initial states as $\gamma (0)=10^{\circ }$,
$\omega _{x}(0)=0$ $^{\circ}/\mathrm{s}$, $\delta _{x}(0)=0^{\circ}$, and
for $i=2,3$, $\mu_i=0.5$, $\gamma_i=k_{i0}=0.1$, $k_1=5$, $\epsilon=0.1$,
$m_{x}^{\delta _{x}}(\lambda \delta _{x})=2.12(1+0.2\sin t)$, $V=200(1+0.1\cos 2t)$.
The simulation results are
given in  Fig. \ref{fig2}. All of the state variables converge to zero asymptotically.

\begin{figure}[!htb]
  \centering
  \includegraphics[width=1\hsize]{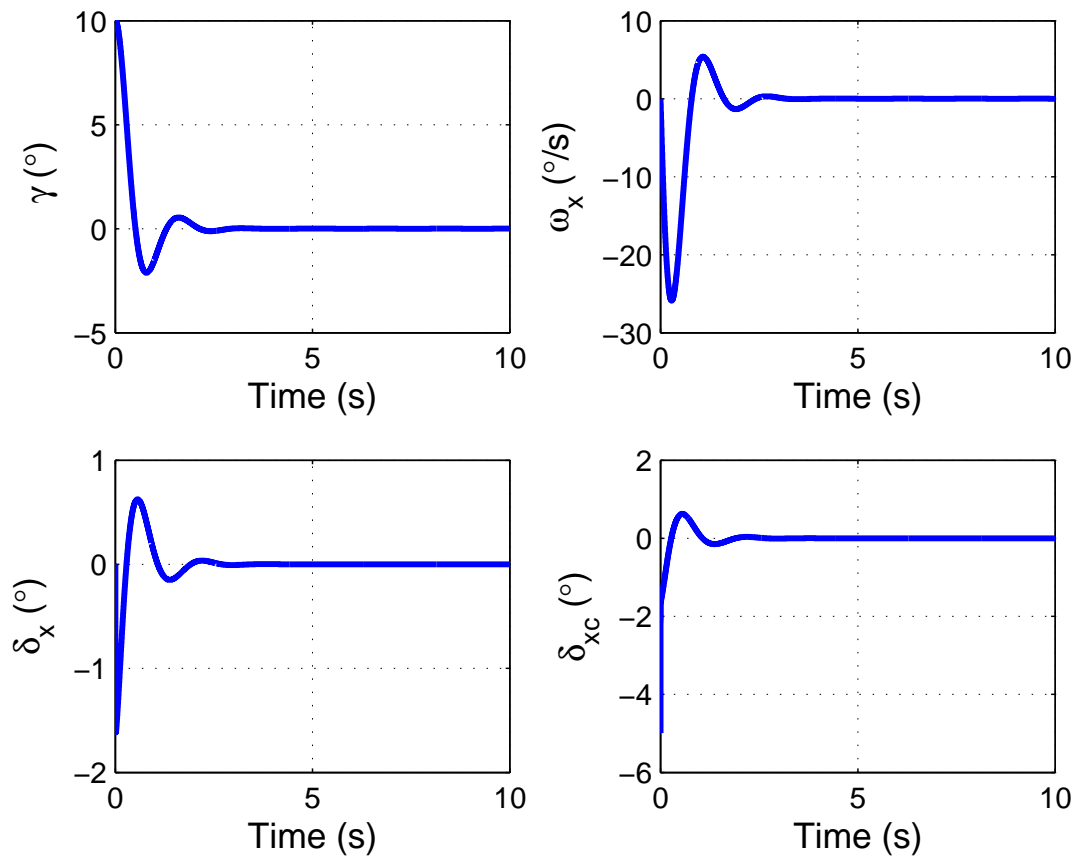}
  \caption{Simulation results for the realistic example}
  \label{fig2}
\end{figure}

The simulation results of the above-mentioned numerical and realistic examples verify the correctness and the effectiveness of the proposed method.

\section{Conclusion}
This paper considers the globally asymptotic state stabilization problem for
pure-feedback systems in the pseudo-affine form and with non-linearly parameterised uncertainties.
Based on the parameter separation technique, a novel adaptive backstepping controller is designed
by utilizing the high gain idea. The proposed controller could guarantee that
all the closed-loop signals are bounded for any initial system condition, and that the state is globally asymptotically stabilized.
A numerical and a realistic examples are given to show the correctness and effectiveness of the proposed approach.

\section{Acknowledgement}
This work was supported by National Natural Science Foundation of China (No.61203125,61503100),
China Postdoctoral Science Foundation (No. 2014M550189), Heilongjiang Postdoctoral Fund (No. LBHZ13076)
and the Fundamental Research Funds for the Central Universities (No. HIT.IBRSEM.A. 201402).

\end{document}